\documentclass[11pt,b5paper,twoside, headrule]{amsart}

\usepackage{amsfonts, amsmath, amssymb,latexsym}
\usepackage{epsfig}
\usepackage[curve]{xy}
\usepackage{algorithmic}
\usepackage{algorithm}
\usepackage{enumerate}
\usepackage{framed}
\usepackage{hyperref}

\footskip=10mm\parskip 0.2cm\addtocounter{page}{0}
\newtheorem{thm}{Theorem}[section]
\newtheorem{cor}[thm]{Corollary}
\newtheorem{lem}[thm]{Lemma}
\newtheorem{prop}[thm]{Proposition}
\theoremstyle{definition}
\newtheorem{defn}[thm]{Definition}

\newtheorem{example}[thm]{Example}
\newtheorem{notation}[thm]{Notation}
\theoremstyle{remark}
\newtheorem{rem}[thm]{Remark}
\numberwithin{equation}{section}

\newcommand{\Cc}{\mathcal{C}}

\newcommand{\R}{{\mathbb R}}
\newcommand{\C}{{\mathbb C}}

\newcommand{\SL}{\operatorname{SL}}
\newcommand{\GL}{\operatorname{GL}}
\newcommand{\Id}{\operatorname{Id}}
\newcommand{\ext}{\operatorname{ext}}
\newcommand{\inte}{\operatorname{int}}
\newcommand{\Tr}{\operatorname{Tr}}

\title{Fuchsian Codes with Arbitrarily High Code Rates}

\begin{document}
%=========================== Baselineskip ===============================
\renewcommand\baselinestretch{1.2}
\renewcommand{\arraystretch}{1}
\def\base{\baselineskip}
%========================= Begin Document--Fonts--and other ============================
\font\tenhtxt=eufm10 scaled \magstep0 \font\tenBbb=msbm10 scaled
\magstep0 \font\tenrm=cmr10 scaled \magstep0 \font\tenbf=cmb10
scaled \magstep0

%%================================head--document=======================

\def\evenhead{{\protect\centerline{\textsl{\large{I. Blanco}}}\hfill}}

\def\oddhead{{\protect\centerline{\textsl{\large{On the non vanishing of the cyclotomic $p$-adic $L$-functions}}}\hfill}}

\pagestyle{myheadings} \markboth{\evenhead}{\oddhead}

\thispagestyle{empty}

\author[I. Blanco-Chac\'{o}n]{Iv\'{a}n Blanco-Chac\'{o}n}
\address{Department of Mathematics and Systems Analysis\\
Aalto University\\
Otakaari 1, M\\
FI-00076 Espoo, Finland}
\email{ivan.blancochacon@aalto.fi}

\author[C. Hollanti]{Camilla Hollanti}
\address{Department of Mathematics and Systems Analysis\\
Aalto University\\
Otakaari 1, M\\
FI-00076 Espoo, Finland}
\email{camilla.hollanti@aalto.fi}

\author[M. Alsina]{Montserrat Alsina}
\address{Universitat Polit\`{e}cnica de Catalunya- BarcelonaTech, Dept. Applied Mathematics III - EPSEM,\\
University of Barcelona\\
Av. Bases de Manresa 61-73\\
08242  Manresa (Spain)}
\email{montserrat.alsina@upc.edu}

\author[D. Rem\'{o}n]{Dion\'{i}s Rem\'{o}n}
\address{Faculty of Mathematics\\
University of Barcelona\\
Gran Via de les Corts Catalanes, 585\\
08007 Barcelona, Spain}
\email{dremon@ub.edu}

\maketitle

\begin{abstract} Recently, so-called Fuchsian codes have been proposed in [I. Blanco-Chac\'on et al., ``Nonuniform Fuchsian codes for
noisy channels'', J. of the Franklin Institute 2014]  for communication over channels subject to additive white Gaussian noise (AWGN). The two main advantages of Fuchsian codes are their ability to compress information, i.e., high code rate, and their logarithmic decoding complexity.  In this paper, we improve the first property further by constructing Fuchsian codes with arbitrarily high code rates while maintaining logarithmic decoding complexity. Namely, in the case of Fuchsian groups derived from quaternion algebras over totally real fields we obtain a code rate that is proportional to the degree of the base field. In particular, we consider  arithmetic Fuchsian groups of signature $(1;e)$ to  construct explicit codes having code rate six, meaning that we can transmit six independent integers during one channel use.
\end{abstract}

\section{Introduction}

Nonuniform codes are known to be good in terms of approaching the capacity of a channel affected by additive white Gaussian noise (AWGN). They have been used already in early-state signal transmission, \emph{e.g.}, in the so-called  \emph{codec} transmission, and are present more recently in the Digital Video Broadcasting Next Generation Handheld (DVB-NGH) standard \cite{dvb}. Unfortunately, nonuniform codes are in general subject to brute-force maximum-likelihood (ML) decoding methods, resulting in a linear decoding complexity in the codebook size. Recently, in \cite{BHAR} and \cite{BHR}, a family of nonlinear and nonuniform \emph{Fuchsian codes} were constructed based on Fuchsian groups of the first kind defined from quaternion algebras over $\mathbb{Q}$.
The decoding procedure of these codes is based on a modified point reduction algorithm having logarithmic complexity in the codebook size, which was shown  to imply logarithmic decoding complexity for the Fuchsian codes \cite{BHAR}, \cite{bayerremon}.

In this paper, we will study one of the key features of a \emph{code}, namely its ability to carry information. In other words, how many bits per channel use we will be able to send when using a given code. Here, we will do this in conjunction with Fuchsian codes. To this end, let us start with the following intuitive definition, which will be formalized later (cf. Def. \ref{defncoderate}).

\begin{defn}
Let $C\subset\mathbb{C}$ denote a codebook of size $|C|<\infty$, and let $k$ denote the number of independent integers embedded in each codeword. The \emph{code rate} is
$$
R=k
$$
 (independent integer) symbols per channel use (spcu).

The \emph{data rate} is
$$
R_d=\log_2|C|
$$
 bits per channel use (bpcu).

\end{defn}

The code rate can be thought of as the ability of the code to compress information, and the data rate as the transmission speed enabled by the code.

\begin{rem}
%The above definition of code rate is not to be confused with the rate of an error correcting code, where the rate describes the level of redundancy. Moreover, our code consists of complex numbers, and hence transmitting each codeword takes only one channel use.
When encoding over multiple channel uses, as is the case for lattice codes and space--time codes \cite{OV}, the suitable definition for code rate is $R=k/T$, where $T$ is the number of channel uses. With $T=1$, this coincides with the definition above.
\end{rem}

Apart from the sub-linear decoding complexity, another advantage of the Fuchsian codes in \cite{BHAR} is that they allow, in the sense of the above definition, to compress information. Namely, they enable us to embed three independent integers in one complex number to be transmitted,  having thus  code rate $R=3$. In comparison to the usual way of transmitting a complex signal, \emph{e.g.}, by using the quadrature amplitude modulation (QAM) consisting of a finite subset of Gaussian integers
$$
\mathbb{Z}[i]=\{c=a+bi\, |\, a,b\in\mathbb{Z}\}
$$
embedding only two independent integers, the rate of a Fuchsian code \cite{BHAR} is 50\% higher. The algebraic reason for the higher rate will become evident in Section 3.

In what follows, we will consider quaternion algebras defined over totally real field extensions of $\mathbb{Q}$ of degree bigger than $1$. As we will see, this implies that we can increase the code rate even further. Since the point reduction algorithm works in general for arithmetic Fuchsian groups,
%defined over totally real number fields splitting exactly at one archimedean place,
we have adapted it to some explicit groups derived from the arithmetic Fuchsian groups of signature $(1;e)$. These groups allow us to construct Fuchsian codes with higher rates. We will show explicit examples of a rate six code, and our method can indeed produce fully explicit codes of rate up to $18$.

The motivation for increasing the code rate initially came for lattice coding \cite{OV}. For lattice codes, higher code rate typically implies higher data rates (or equivalently, bigger codebook, cf. Def. 1.1.) without having to increase the transmission power or to compromise the minimum distance. As the transmission power is determined by the Euclidean norm of the transmitted codeword, this results from the fact that a higher rank lattice with a unit volume has more points within a Euclidean hypercube of a given edge length than a lower rank lattice with a unit volume. To get an intuition, one can think of how many integer points are there in the real line between, say, 0 and 10, versus how many integral  points are there in a $10\times 10\times 10$ cube in $\mathbb{R}^3$. The same is naturally valid for Euclidean hyperspheres. Therefore, it is desirable to maximize the code rate. For the proposed higher rate Fuchsian codes there is a caveat: due to the nonlinear and nonuniform structure, there is no \emph{a priori} reason why higher rate should imply a bigger codebook. It seems difficult to give a rigorous proof for this, so we have settled with numerical experiments to see how the rate affects the codebook size. In our example cases, higher rate seems to indeed imply a bigger codebook, given the minimum distance and the hypersphere radius.

Our interest in Fuchsian groups as a basis for code construction stems from a series of recent papers by Palazzo \emph{et al.} In \cite{brazilian_IEEE,brazilian_COAM,brazilian_franklin2,brazilian_franklin}, among others, various interesting connections between Fuchsian groups and signal constellation design are presented.  In \cite{brazilian_IEEE}, the authors construct Fuchsian groups suitable for signal constellation design.  In  \cite{brazilian_franklin2}, the authors consider the unit disk model  of the hyperbolic plane as the signal space, and the noise is modeled as a hyperbolic Gaussian random variable. By using some results of hyperbolic geometry they construct a hyperbolic equivalent to  QAM and PSK constellations and point out that, when the channel model is hyperbolic (this is the case \emph{e.g.} in power transmission line communications \cite{59paper}), the proposed hyperbolic constellations provide higher coding gains than the classical Euclidean variants. Building on this work, in \cite{brazilian_franklin} the authors construct dense tessellations and count Dirichlet domains attached to certain families of these tessellations.  In \cite{brazilian_COAM} the authors use units of quaternion orders to construct space-time matrices with the potential use case being wireless multi-antenna (MIMO) communications. We refer the reader to \cite{SRS,maxorder} as the early references to the use of division algebras and maximal orders in MIMO.

Although codes related to Fuchsian groups had been considered before, our approach in \cite{BHAR} was original in that it described a complete construction and decoding process, whereas earlier work had largely  concentrated on the constellation design while giving little attention to the decoding and performance aspects.
Another key difference to the aforementioned works was, as is the case of the present paper,  that we are studying codes on the \emph{complex plane} arising from quaternion algebras and Fuchsian groups, and  our aim is to apply the codes to the classical  (Euclidean) channel models such as the aforementioned AWGN channel, with possible future extension to fading  channels. We do not use hyperbolic metric as our design metric, but use the Fuchsian group as a starting point to the  code generation. Nevertheless, our decoder will rely on hyperbolic geometry as opposed to the classical decoders based on Euclidean geometry.

The paper is organized as follows: in Section 2 we  give some background and notation on Fuchsian groups acting on the complex upper half-plane, specially those coming from quaternion algebras over a number field $F$, in a more general setting than \cite{BHR} and \cite{BHAR}.
In Section 3,  we generalize the construction of Fuchsian codes in order to obtain codes of arbitrarily high rates. In particular, by using quaternion algebras over totally real extensions $F/\mathbb{Q}$, we prove that the code rate is at least $3n$, where $n=[F:\mathbb{Q}]$ is the degree of the base field. In Section $4$, we explore such Fuchsian codes for a number of Fuchsian groups derived from those of signature $(1;e)$, classified by Takeuchi \cite{tak}. This task has required to explicitly construct suitable fundamental domains for these groups and to adapt the point reduction algorithm \cite{bayerremon} to our case. We also provide some numerical evidence to justify the study of higher rate Fuchsian codes by showing that increasing the rate may indeed increase the codebook size and hence the data rate.
In the last section, we expose the conclusions and discuss directions for future research.

\section{Action of arithmetic Fuchsian groups on the hyperbolic plane}

Next we review some algebraic concepts and results related to Fuchsian groups acting on the complex upper half-plane, in particular those arising from quaternion algebras over a number field $F$, thus extending our work in \cite{BHAR}.

\subsection{Tessellations on the complex plane}

Let us denote by $\mathrm{SL}(2,\mathbb{R})$ the special linear group of  $2\times 2$-matrices with entries in $\mathbb{R}$ and determinant equal to 1. There is a group action on the Riemann sphere $\mathbb{C}\cup \{\infty\}$ defined by:

\begin{equation}
\gamma(z)=\dfrac{a z + b}{c z + d}, \quad \gamma(\infty)=\dfrac{a}{c}=\lim_{z\rightarrow \infty} \gamma(z), \quad \forall \gamma=\left( \begin{array}{cr} a & \quad b \\ c & d \end{array}\right)\in \mathrm{SL}(2,\mathbb{R}), \ z\in \mathbb{C}.
\end{equation}

These maps $z\mapsto \gamma(z)$ are called \emph{fractional linear transformations} or \emph{M\"{o}bius transformations} of the Riemann sphere. The unique matrices fixing all the points are $\Id$ and $-\Id$. Thus,  $\mathrm{SL}(2,\mathbb{R})/\{\pm \Id\}$ acts faithfully on $\mathbb{C}$, that is to say, each element other than the identity acts nontrivially.
The complex upper half-plane $$\mathcal{H}=\{z\in \mathbb{C} \mid \Im (z)>0\}$$ is stable under this action. In fact the M\"{o}bius transformations are the group of isometries of $\mathcal{H}$ with respect to the hyperbolic geometry.

We will consider discrete subgroups of $\mathrm{SL}(2, \mathbb{R})$ with a proper and discontinuous action on $\mathcal{H}$ such that the hyperbolic volume of the quotient of $\mathcal{H}$ by that action is finite. These groups are called Fuchsian groups of the first kind, Fuchsian groups in short. By abuse of notation, we will use the same notation for the groups in  $\mathrm{SL}(2,\mathbb{R})$ and   $\mathrm{SL}(2,\mathbb{R})/\{\pm \Id\}$. Let us also recall that whenever a group acts on a set, it divides the set into equivalence classes.

\begin{defn}
For a Fuchsian group $\Gamma$, a fundamental domain is a closed hyperbolic polygon $\mathcal{F}$ in $\mathcal{H}$ satisfying:
\begin{itemize}
\item[a)] For any $z,z'$ in the interior of $\mathcal{F}$, if there exists $\gamma\in\Gamma$ such that $\gamma(z)=z'$, then $z=z'$ and $\gamma=Id$.
\item[b)] For any $z\in\mathcal{H}$, there exists $z_0\in\mathcal{F}$ and $\gamma\in\Gamma$ such that $\gamma(z)=z_0$.
\end{itemize}
\end{defn}

Each election of a fundamental domain for the action of a Fuchsian group $\Gamma$ leads to a regular tessellation of the upper half-plane by hyperbolic polygons, which will be useful for the code construction. In fact the unlimited number of tessellations is one of the advantages of the hyperbolic plane compared to the Euclidean one.

Given a fundamental domain of $\Gamma$ and $z\in\mathcal{H}$, the problem of finding $z_0\in\mathcal{F}$ and $\gamma\in\Gamma$ such that $\gamma(z)=z_0$ is known as the \emph{point reduction problem} and requires an algorithmic solution referred to as the point reduction algorithm. Here, we will consider tessellations obtained from Fuchsian groups arising from quaternion algebras.

\subsection{Quaternion algebras, orders and arithmetic Fuchsian groups} \label{sec-alg}

Let $F$ be an arbitrary field with $\operatorname{char} F\neq 2$. The quaternion algebra $H$ denoted by $\left(\frac{a,b}{F}\right)$,  $a,b\in F\smallsetminus \{0\}$, is the $F$-algebra with $F$-basis $\{1,I,J,K\}$ subject to the multiplication rules
$I^2=a,J^2=b,K=IJ=-JI$. For more details on quaternion algebras, cf. \cite{vigneras} and \cite{alsinabayer}.

In a quaternion algebra, there is a natural conjugation such that for $\omega= x+yI+zJ+tK\in H$ the conjugate is $\overline{\omega}=x-yI-zJ-tK$. Then the reduced trace and the reduced norm are defined by
$$
\mathrm{Tr}(\omega)=\omega+\overline{\omega}=2x, \qquad
\mathrm{N}(\omega)=\omega\overline{\omega}=x^2-ay^2-bz^2+abt^2.
$$
It is well known that a quaternion $F$-algebra $H$ is a central simple algebra of dimension $4$ over $F$, and by Wedderburn's structure theorem $H$ is isomorphic to either $\mathrm{M}(2,F)$ or to a skew field, also called a division $F$-algebra. If $F=\mathbb{C}$, or more generally $F$ is algebraically closed, only matrix algebras are obtained. If $F=\mathbb{R}$, or in general a local field different from $\mathbb{C}$, there exists a unique division $F$-algebra up to isomorphism. In the real case, the unique quaternion division algebra is the Hamilton quaternion $\mathbb{R}$-algebra, $\mathbb{H}=\left(\frac{-1,-1}{\mathbb{R}}\right)$. Of course, for any field $F$,  $\left(\frac{1,1}{F}\right)\cong \mathrm{M}(2,F)$.

Given $F$ a number field, for each place  $\nu$ of $F$, that is an archimedean or non-archimedean absolute value of $F$, consider the local field $F_{\nu}$. Then $H_{\nu}:=H\otimes F_{\nu}$ is a quaternion algebra over the local field. We say that $H$ is \emph{ramified} at $\nu$ if $H_{\nu}$ is a division algebra; otherwise, $H$ is said to \emph{split} at $\nu$. The set of places where $H$ is ramified is a finite set of even cardinality and it characterizes the quaternion algebra up to isomorphism.

From now on, consider $F$  a totally real algebraic number field with ring of integers $R$ and $[F:\mathbb{Q}]=n$. We will assume that $H$ is a division algebra, that is $H\ncong \mathrm{M}(2,F)$, and that $H$ ramifies precisely at $n-1$ out of the $n$ completions of $F$ with respect to the Galois embeddings of $F/\mathbb{Q}$ into $\mathbb{R}$, the archimedean places.  Namely, $H$ satisfies the following condition
\begin{equation}\label{cond-quat-alg}
H\otimes_{\mathbb{Q}} \mathbb{R} \cong\mathrm{M}(2,\mathbb{R})\times \mathbb{H}^{n-1}.
\end{equation}

In the case of quaternion $\mathbb{Q}$-algebras, this means that $H$ is an indefinite quaternion algebra. Such algebras were used to consider Fuchsian groups and to define associated Fuchsian codes in \cite{BHAR}.

\begin{lem} \label{reg-rep}
The following map $\phi$ is a monomorphism of $F$-algebras, giving a left regular representation of the quaternion algebra in a matrix algebra:
$$
\begin{array}{ccc}
\phi: \left(\dfrac{a,\, b}{F}\right) & \to & \mathrm{M}(2,\left(F(\sqrt{a})\right))\\
x+yI+zJ+tK & \mapsto &
\left(\begin{array}{ccc} x+y\sqrt{a} &\phantom{x} & z+t\sqrt{a}\\
b(z-t\sqrt{a})&\phantom{x} & x-y\sqrt{a}\end{array}\right).
\end{array}
$$
\end{lem}

\begin{rem}
Notice that for any $\omega\in H$, $\mathrm{N}(\omega)=\mathrm{det}\left(\phi(\omega)\right)$, and $\mathrm{Tr}(\omega)=\mathrm{Tr}\left(\phi(\omega)\right)$. If $\sqrt{a}\in \mathbb{R}$, then $\mathrm{M}(2,\left(F(\sqrt{a})\right))\subseteq \mathrm{M}(2,\mathbb{R})$. In particular, if we restrict to quaternion elements in $H$ with reduced norm equal to 1, then the image under $\phi$ is contained in $\mathrm{SL}(2,\mathbb{R})$.
\end{rem}

An $R$-order $\mathcal{O}$ in $H$ is a finitely generated $R$-submodule, \emph{i.e.},  a subring such that $\mathcal{O}\otimes F\simeq H$. The elements in $\mathcal{O}$ are integral, namely for any $\omega\in\mathcal{O}$, its characteristic polynomial  $x^2-\mathrm{Tr}(\omega) x + \mathrm{N}(\omega)$ has coefficients in $R$.

For an order $\mathcal{O}$,  a Fuchsian group $\Gamma$ is obtained by using the map $\phi$:
$$
\Gamma=\phi(\mathcal{O}^*_{+}), \qquad \text{where } \mathcal{O}^*_{+}=  \{ \omega\in\mathcal{O} \,\mid \, \omega \mbox{ invertible}, \, \mathrm{N}(\omega)>0\}.
$$

In the case of the base field being $\mathbb{Q}$, it is interesting to consider Eichler orders $\mathcal{O}$. Eichler orders are intersections of two maximal orders. When $H$ is indefinite, the discriminant of the quaternion algebra $D\in \mathbb{Z}$ is defined as the product of all finite primes where $H$ is ramified. The discriminant  is an invariant of the quaternion algebra. The group $\phi(\mathcal{O}^*_{+})$ is denoted by $\Gamma(D,N)$,  where $N\in \mathbb{N}$ is the level of the Eichler order, and $N=1$ for a maximal order. The group $\Gamma(D,N)$ is well-defined up to conjugation.

\begin{rem}
After Weil's results on the classification of classical groups, the list of all arithmetic subgroups of $\SL(2,\mathbb{R})$  is exhausted up to commensurability by Fuchsian groups coming from quaternion algebras over totally real number fields, cf. \cite{katok}, where two groups $G_1$ and $G_2$ are said to be commensurable if $G_1\cap G_2$ has finite index both in $G_1$ and in $G_2$. Thus the Fuchsian groups derived from quaternion algebras are in the main focus when studying tessellations.
\end{rem}

Fuchsian groups derived from quaternion algebras and their quotients $\Gamma \backslash \mathcal{H}$ also led to the theory of Shimura curves in the sixties \cite{shimura1967}. The case of matrix algebras corresponds to classical modular curves. Since we assumed that $H$ is a division algebra, the quotient is already compact and there are no cusps. Thus the classical problems of finding fundamental domains and reducing points to a given fundamental domain
call for algorithmic solutions different from those available for the modular case. For more results on fundamental domains, see \cite{alsinabayer}, \cite{johansson}, \cite{voight}.

As already mentioned, a general algorithm for the point reduction problem was recently proposed in \cite{bayerremon}. Some specific  examples  can  be found in \cite{BHAR}.

\subsection{The point reduction algorithm (PRA)}

As the decoding of Fuchsian codes is based on the point reduction algorithm, let us summarize here its main features. See \cite{bayerremon} for validity and complexity proofs.

First, let us recall that the construction of a fundamental domain for $\Gamma$ following Fords's method (cf. \cite{katok}) is based on the use of \emph{isometric circles} $I(\gamma)$, \emph{i.e.}, hyperbolic lines in the upper half-plane associated to the  matrices $\gamma\in \Gamma$:
$$
I(\gamma)= \{z\in\mathcal{H} \,\mid \,  |c z+d| =1\},  \quad
\text{ for } \gamma=\left( \begin{array}{cr} a & \quad b \\ c & d \end{array}\right)\in\Gamma, c\not=0.
$$

\begin{notation}\label{Mdef} \label{rem-G}
For any Fuchsian group $\Gamma$ and any fixed fundamental domain $\mathcal{F}(\Gamma)$, let us denote by $G$ the minimal subset of $\Gamma$ such that the edges of $\mathcal{F}(\Gamma)$ are included in the set of isometric circles defined by the elements of $G$. As a presentation of the group $\Gamma$ arises from the pairing of the edges, we can assume that the generators of $\Gamma$ are included in $G$.
The set $G$ splits into two subsets denoted by $G^{\text{int}}$ and $G^{\text{ext}}$ according to whether the fixed fundamental domain is located in the interior or in the exterior of each isometric circle, respectively. Hence, if $\ext(I(\gamma))$ and $\inte(I(\gamma))$  denote the exterior and the interior of the isometric circle $I(\gamma)$,  the fundamental domain $\mathcal{F}(\Gamma)$ is the closure of
 $$ \bigcap_{\gamma\in G^{\text{ext}}} \ext(I(\gamma)) \cap  \bigcap_{\gamma\in G^{\text{int}}} \inte(I(\gamma)).$$
\end{notation}

Now we are ready to introduce the point reduction algorithm (PRA). It gives a solution to the reduction point problem; namely, it reduces a given point $z\in \mathcal{H}$ to a point $z_0\in \mathcal{F}$, and yields a transformation $\gamma\in \Gamma$ such that $\gamma(z)=z_0$.

\begin{algorithm}\label{algorithm}
\textbf{PRA (Point Reduction Algorithm)}\\
  \begin{itemize}
\item[]  \textbf{Step 1} Initialize: $z_0 = z$ and $\gamma = \mathrm{Id}$.
\item[]  \textbf{Step 2} Check if $z_0 \in \mathcal{F}$.
         \newline \hphantom{\textbf{Step 2\,}}
         If $z_0 \in \mathcal{F}$, return $z_0$ and $\gamma$. Quit.
         \newline \hphantom{\textbf{Step 2\,}}
        If $z_0 \not\in \mathcal{F}$, return $g\in G$ such that:
        \newline \hphantom{\textbf{Step 2\,}}
        \hspace{0.5cm}$z_0 \in \mathrm{int}(\mathrm{I}(g))$, if $g\in G^{\text{ext}}$,
                \newline \hphantom{\textbf{Step 2\,}}
         \hspace{0.5cm}$z_0 \in \mathrm{ext}(\mathrm{I}(g))$ if $g\in G^{\text{int}}$.
\item[]  \textbf{Step 3}    Compute $z_0 = g (z_0)$ and $\gamma = g\cdot \gamma$. Go to Step 2.
  \end{itemize}
\end{algorithm}

The PRA will be used for the decoding of Fuchsian codes in the sequel. In order to study the complexity of the algorithm when applied to Fuchsian codes, we state the following remark and definition.
\begin{rem}
Consider $z,z'\in \gamma(\mathcal{F})$, $\gamma\in\Gamma$, $z\not=z'$. By construction, to run the PRA  with input $z$ or $z'$ will output different points $z_0$ or $z'_0$, respectively, but the same matrix $g=\gamma$, in the same number of steps.
\end{rem}

\begin{defn}\label{depth}
Given a matrix $\gamma\in\Gamma$, the \emph{depth} of $\gamma$, denoted by $\ell(\gamma)$, is the minimal number of iterations of the PRA  to reduce $\gamma(\tau)$ to the fundamental domain for any $\tau\in \mathcal{F}$.
\end{defn}

\section{General construction of Fuchsian codes}
In \cite{BHAR}, we described in detail how to construct and decode Fuchsian codes in an AWGN channel. In order to make this paper self-contained,  we shortly restate the process in this more general setting, focusing on some essential properties.

\subsection{General construction}

Let $\Gamma$ be a Fuchsian group as in Section 2.1.

The first step in the construction of the code is to fix a fundamental domain $\mathcal{F} = \mathcal{F}(\Gamma)$, which determines a tessellation of the complex upper half-plane $\mathcal{H}$, and a set $G\subset \Gamma$ whose corresponding isometry circles are the edges of the fundamental domain, following the notation in Section 2.

The main step in the process is to choose a set of $N$ different elements in $\Gamma$, $S_{\Gamma}=\{ \gamma_1, \ldots \gamma_N \}$, what is equivalent to choosing $N$ different tiles in the tessellation.
Moreover,  we choose $\tau$ to be an interior point of  $\mathcal{F}$; this condition ensures that $\gamma(\tau)\not=\tau$ for all $\gamma\in \Gamma\setminus \{\pm\Id\}$.

Finally, considering the action of the group $\Gamma$ in the complex upper half-plane $\mathcal{H}$, we obtain the codewords  $\gamma_1(\tau), \ldots \gamma_N(\tau)$ in $\mathcal{H}$.  The condition on $\tau$ ensures $\gamma_i(\tau)\not=\gamma_j(\tau)$. We can double the number of points by expanding to the lower half-plane in a natural way by including the opposites $-\gamma_i(\tau)$. Thus, the codebook consists of the $2N$ complex points constructed by using $\tau$ and $S_{\Gamma}$, and the symmetry with respect to the origin.  Based on the outlined process, we give a formal definition of a Fuchsian code below, and summarize the construction process  in Table 1.

\begin{defn}
Let  $\Gamma$ be a Fuchsian group. Given a fundamental domain $\mathcal{F}(\Gamma)$, a set $S_{\Gamma}=\{ \gamma_1, \ldots \gamma_N \}\subset \Gamma$,  and a point $\tau$ in the interior of $\mathcal{F}(\Gamma)$, we define the associated \emph{Fuchsian code} as $\mathcal{C} = \left\{\pm\gamma(\tau) \, \mid \,  \gamma \in S_\Gamma\right\}\subseteq \mathbb{C}$. The set of codewords is also referred as a $q$-\emph{nonuniform Fuchsian constellation},  $q$-NUF in short, where $q=|\mathcal{C}|=2N$ is the size of the code. The point $\tau$ will be called the \emph{center} of the code.
\end{defn}

\begin{table}[H]\label{construction-figure}
\caption{Sketch of the code construction process.}
$
\boxed{
%\vglue 1truecm
\begin{array}{ccc}
\boxed{
\begin{array}{c}
\text{FIX:} \\
 \text{\small Fuchsian group }  \Gamma\\
\text{\small Fund.\,domain } \mathcal{F}(\Gamma)
\end{array}
      }
&
\rightarrow
\boxed{
\begin{array}{c}
\text{CHOOSE:} \\
S_{\Gamma}=\{\gamma_1, \ldots \gamma_N \}\subset \Gamma \\
\tau \text{\small \ in the interior of } \mathcal{F}(\Gamma)
\end{array}
}
&
\rightarrow
\boxed{
\begin{array}{c}
\text{SET CODEBOOK }\mathcal{C}: \\
\{\pm\gamma_1(\tau), \ldots \pm\gamma_N(\tau)\} \\
|\mathcal{C}|=2N
\end{array}
}
\end{array}
}
$
\end{table}

This construction was stated in \cite{BHAR} in the case of  groups $\Gamma(D,1)$ derived from quaternion algebras  over $\mathbb{Q}$. We refer the interested reader there for explicit examples, including details about representations of the groups,  fundamental domains,  centers of the codes, lists of codewords, as well as some experimental results.

The general construction stated above now allows us to construct more  general codes, as long as we are able to determine the respective fundamental domains. In particular, it can be applied to groups derived from quaternion algebras over a totally real field.  The behavior and performance of Fuchsian codes will essentially depend  on the choices in the intermediate step. An algebraic and geometric study of Fuchsian groups and their fundamental domains will therefore be useful for developing a general understanding of various code parameters, such as the minimum distance and average transmission power. We refer again to \cite{BHAR} for a more detailed exposition.

\subsection{Decoding of Fuchsian codes}

Let $\mathcal{C}\subseteq \mathbb{C}$ be a $q$-NUF constellation with center $\tau$, associated to a fixed Fuchsian group $\Gamma$ with a fixed fundamental domain $\mathcal{F}$.
It is clear that given $x\in\mathcal{C}$, $\Im(x)>0$,  the PRA described in Section 2.3 computes $\gamma\in \Gamma$ such that $x=\gamma(\tau)$, which  is equivalent to finding the tile containing $x$ in the tessellation of $\mathcal{H}$ induced by the fundamental domain $\mathcal{F}$. If $\Im(x)<0$, it is enough to consider $-x$ and then apply  the PRA.

In the context of AWGN channels, let $x=\gamma(\tau)\in \Cc\subset\C$ be the transmitted codeword  and $y$ the  received signal,  $y = x + \varepsilon=\gamma(\tau) + \varepsilon\in\mathbb{C}$, where $\varepsilon$ is the Gaussian noise. The basic  idea underlying our decoding technique is that, provided that the channel is of sufficiently good quality, the received signal $y$ will belong to  the tile $\gamma(\mathcal{F})$ determined by $x$. In other words, when we apply the PRA  to $y$, it will return $\gamma$ and  the transmitted codeword $x=\gamma(\tau)$ can be recovered. In order to measure the decoding complexity when employing the PRA, we define (also cf. Def. \ref{depth}):

\begin{defn}
The \emph{depth} of the code is $\ell(\mathcal{C}):=\max\{\ell(\gamma) \mid \gamma\in S_\Gamma\}$.
\end{defn}

Next, we describe the encoding and decoding process of Fuchsian codes in detail.  In order to remain in the upper half-plane whilst decoding, we initialize the algorithm with $z_0=y$ if $\Im(y)>0$, and with $z_0=-y$ if $\Im(y)<0$. Since $\R$ has measure zero in $\C$, the case $\Im(y)=0$ occurs with probability zero.
%Moreover it is very difficult to obtain a perfect zero from communication channels.

\begin{algorithm}[h!]\label{algorithm}
\textbf{Encoding and decoding of Fuchsian codes}\\
  \begin{itemize}
\item[]  \textbf{Step 1} Assign a matrix  $\gamma \in \mathbb{C}$.
\item[]  \textbf{Step 2} Compute the codeword $x = \gamma(\tau)$.
\item[]  \textbf{Step 3} Transmit $x$ using the AWGN channel.
\newline \hphantom{\textbf{Step 2\,}} The receiver obtains $y = x + \varepsilon$,  where $\varepsilon$ is the Gaussian noise.
\item[]\textbf{Step 4} Decode the signal $y$:
\newline \hphantom{\textbf{Step 2\,}}
         If $\Im(y)>0$, apply PRA to y, obtain $\gamma$.
\newline \hphantom{\textbf{Step 2\,}}
         If $\Im(y)<0$, save the sign information,  $y \leftarrow -y$,
\newline \hphantom{\textbf{Step 2\,}} \hphantom{ If $\Im(y)<0$} apply PRA to $y$, obtain $\gamma$.
\item[]  \textbf{Step 5} $\gamma \leftarrow \text{ (sign information)}\times \gamma$.
\end{itemize}
\end{algorithm}

The following theorem proves the existence of Fuchsian codes with logarithmic decoding complexity.

\begin{thm}
Let $\Gamma$ be a Fuchsian group containing a non-abelian free subgroup. There exist Fuchsian codes $\mathcal{C}$ associated to $\Gamma$ such that the decoding algorithm for $\mathcal{C}$ has logarithmic complexity in $|\mathcal{C}|$, namely,  the number $r_{\mathcal{C}}$ of arithmetic operations satisfies
$$
	r_{\mathcal{C}} = O(\log(|\mathcal{C}|).
$$
\end{thm}

The proof is analogous to the corresponding results in \cite{BHAR}, so we only provide a sketch of the proof. The first part of the proof is to count the maximal number of arithmetic operations when running the PRA. In each iteration, we got that the number of operations only depends on the fundamental domain, and not on the code size. As the number of iteration is bounded by the depth $\ell(\mathcal{C})$,  $r_{\mathcal{C}}=O(\ell(\mathcal{C}))$.
Secondly, by using the technical condition given by the existence of a non-abelian free subgroup, a Fuchsian code $\mathcal{C}$ can be constructed in such a way that $\ell(\mathcal{C})=O(\log(|\mathcal{C}|)$. The key point is to choose $S_{\Gamma}$ as large as possible while controlling the depth of their elements.

Combining these two parts, we deduce the existence of Fuchsian codes with logarithmic complexity.

Actually, for a fixed Fuchsian group $\Gamma$ the complexity of the decoding algorithm depends only on the selection of the subset $S_{\Gamma}$. The choice of the center of the code $\tau$ will influence the performance of the code, being related to the minimum border distance, as stated in \cite{BHAR} (see the code design criterion therein).

\subsection{The rate}
Let us reduce to the case of Fuchsian groups derived from quaternion algebras $H$ over $F$, as in Section \ref{sec-alg}. To this end, let $F$ be a totally real number field  with ring of integers $R$ and $[F:\mathbb{Q}]=n$. Now $H$ is a division algebra satisfying condition \ref{cond-quat-alg}, and $\Gamma=\phi(\mathcal{O}^*_{+})$ for an order $\mathcal{O}$ in $H$, and $\phi$ the regular representation of $H$ in $\mathrm{M}(2,\mathbb{R})$ (cf. Lemma \ref{reg-rep}). Consider a code $\mathcal{C} = \left\{\pm\gamma(\tau) \, \mid \,  \gamma \in S_\Gamma\right\}\subseteq \mathbb{C}$.

In the case $F=\mathbb{Q}$, when we restrict to the natural order $\mathbb{Z}[1,I,J,K]$, a complex number $\gamma(\tau)$ to be transmitted is identified with the matrix $\gamma\in S_{\Gamma}\subset \phi(\mathcal{O}^*_{+})$ (recall that we take $\tau$ an interior point of the fundamental domain). Writing $\gamma=\left(\begin{array}{cc}x+y\sqrt{a} & z+t\sqrt{a}\\b(z-t\sqrt{a}) & x-y\sqrt{a}\end{array}\right)$ with $x,y,z,t\in\mathbb{Z}$, we can identify $\gamma$ with the $4$-tuple $(x,y,z,t)\in\mathbb{Z}^4$, which is subject to the normic equation
\begin{equation}\label{normic}
x^2-ay^2-bz^2+abt^2=1,\quad a>0
\end{equation}
Thus,
%Due to the normic equation \eqref{normic},
the $4$-tuple consists of $3$ algebraically independent integers. This is equivalent to say that the set of $4$-tuples satisfying the normic equation has $3$ algebraic degrees of freedom. Notice that this is precisely the dimension of the algebraic set defined by the normic equation (which is not empty, since it contains all the infinite $4$-tuples attached to the Fuchsian group). The concept of algebraic code rate, denoted $R$, will be hence defined so that for this code we have $R=3$ symbols per channel use (spcu). We can easily generalize this notion of code rate for Fuchsian codes over totally real number fields.

Let $F$ be a totally real number field of degree $n$ with ring of integers $R$, let $B$ be a quaternion $F$-algebra satisfying condition \ref{cond-quat-alg} and $\Gamma=\phi\left(\mathcal{O}^*_+\right)$, with $\mathcal{O}$ a maximal $R$-order. Each matrix $\gamma\in\Gamma$ has the form
$\gamma=\left(\begin{array}{cc}x+y\sqrt{a} & z+t\sqrt{a}\\b(z-t\sqrt{a}) & x-y\sqrt{a}\end{array}\right)$ with $(x,y,z,t)\in R^4$. Fixing a $\mathbb{Z}$-basis of $R$, we can identify $x$ with an $n$-tuple $(x_1,...,x_n)\in\mathbb{Z}^n$ and analogously with $y$, $z$, and $t$. Hence, we can identify the matrix $\gamma$ with a $4n$-tuple $(x_1,...,x_n,y_1,...y_n,z_1...z_n,t_1,...,t_n)\in\mathbb{Z}^4$.

The fact that the $4$-tuple $(x,y,z,t)\in R^4$ satisfies the normic equation (defined over $R$) is equivalent to the fact that the corresponding $4n$-tuple satisfies a certain system of polynomial equations (defined over $\mathbb{Z}$). This system of polynomial equations defines an algebraic set which we denote by $A(\Gamma)$. Notice that this algebraic set has infinitely many elements, hence, its algebraic dimension is well defined.

\begin{defn}The algebraic code rate in symbols per channel use (spcu), or code rate from now on, of the Fuchsian code attached to a Fuchsian group $\Gamma=\phi(\mathcal{O}^*_+)$ satisfying condition \ref{cond-quat-alg} is the algebraic dimension of the algebraic set $A(\Gamma)$.
\label{defncoderate}
\end{defn}

\begin{rem}The code rate defined this way, measures how many degrees of freedoms are there in the set of $4n$-tuples attached to the Fuchsian group. We can think of the code rate, hence, as the number of algebraically independent (non-redundant) symbols in each $4n$-tuple, or as we said in the introduction, the maximal number of independent symbols embedded in each codeword $\gamma(\tau)$, with $\gamma\in\Gamma$.
\end{rem}

Notice that if an undetermined system of $t$ polynomial equations in $n$ variables has solutions, then the set of all complex solutions is an algebraic set of dimension at least $n-t$. In particular, the code rate of a Fuchsian code will be at least $4n-t$, being $t$ the number of equations defined over $\mathbb{Z}$ which are equivalent to the normic equation, which is defined over $R$.

The main result on the existence of Fuchsian codes of arbitrarily high rates is the following theorem. The proof consists on proving the two propositions stated after the theorem.

\begin{thm}[Main Theorem]\label{thm-rate}
Let $F$ be a totally real number field of degree $n$. There exist infinitely many Fuchsian codes with rate at least $3n$ attached to $F$.
\end{thm}

In order to prove our main theorem we  first prove that, for a fixed quaternion algebra satisfying \ref{cond-quat-alg}, there are Fuchsian codes of rate at least $3n$. This is Proposition \ref{prop-rate-3n}. Second, we prove the existence of such quaternion algebras in Proposition \ref{ex-quat-alg}.

\begin{prop} \label{prop-rate-3n}
Let $F/\mathbb{Q}$ be a totally real number field of degree $n$, with ring of integers $R$, and $H$ a quaternion $F$-algebra satisfying condition \ref{cond-quat-alg}. Then, a Fuchsian code associated to the natural order $R[1,I,J,K]$ has code rate at least $3n$.
\end{prop}
%% R=\mathbb{Z}[\theta]

\begin{proof}
Consider the natural order $R[1,I,J,K]$ of the quaternion algebra $H=\left(\frac{a,\, b}{F}\right)$. Then $\Gamma$ is determined by  $4$-tuples $(x,y,z,t)\in R^4$ satisfying  $x^2-ay^2-bz^2+abt^2=1$.

Let $\{\theta_1,...,\theta_n\}$ be a $\mathbb{Z}$-basis of $R$ with $\theta_1=1$. Writing $x=\sum_{k=1}^{n}x_k\theta_k$, $y=\sum_{k=1}^{n}y_k\theta_k$, $z=\sum_{k=1}^{n}z_k\theta_k$, $t=\sum_{k=1}^{n}t_k\theta_k$, each of these algebraic integers can be identified with its coordinates in the integral basis. Thus any 4-tuple $(x,y,z,t)\in R^4$ can be identified as a $4n$-tuple of rational integers.
%Consider $\theta\in F$ a primitive element such that $R=\mathbb{Z}[\theta]$.

Let us expand the normic equation $x^2-ay^2-bz^2+abt^2=1$ such that it corresponds to a system of polynomial equations defined over $\mathbb{Z}$. We set  $x^2=\sum_{k=1}^{n}f_{x,k}(x_1,...,x_{n})\theta_{k}$, with $f_{x,k}\in\mathbb{Z}[x_1,...,x_n]$ a quadratic homogeneous polynomial, and analogously for $y,z,t$, and obtain
$$
x^2-ay^2-bz^2+abt^2=\sum_{k=0}^{n-1}(f_{x,k}-af_{y,k}-bf_{z,k}+abf_{t,k})\theta_{k}.
$$
Since $a,b\in R$, the equation can be rewritten as
$$
x^2-ay^2-bz^2+abt^2=\sum_{k=1}^{n}g_k\theta_k,
$$
with $g_k\in\mathbb{Z}[x_1,...,x_n,...,t_1,...,t_n]$ quadratic homogeneous polynomials.
The condition $x^2-ay^2-bz^2+abt^2=1$ now becomes equivalent to
$$
\begin{array}{l}
g_1(x_1,...,x_n,...,t_1,...,t_n)=1,\\
g_k(x_1,...,x_n,...,t_1,...,t_n)=0, \quad \text{for } 2\leq k\leq n.
\end{array}
$$
This system defines the algebraic set $A(\Gamma)$. Since a $4n$-tuple corresponding to an element of the Fuchsian group bears $n$ restrictions, and the algebraic set $A(\Gamma)$ contains infinitely many solutions, we see that the algebraic dimension of $A(\Gamma)$, or equivalently, the code rate of the Fuchsian code attached to $\Gamma$ is at least $3n$ spcu, proving the proposition.
\end{proof}

The following proposition addresses the question whether there exist quaternion algebras to which the above proposition can be applied.

\begin{prop}\label{ex-quat-alg}
Let $F$ be a totally real number field of degree $n$. There exist infinitely many quaternion algebras $H$ over $F$ satisfying the condition \ref{cond-quat-alg}.
\end{prop}

\begin{proof}
If $n\geq 3$ is odd, define $\Sigma$ to be the set of all but one archimedean absolute values of $F$. Otherwise, define $\Sigma$ as the set of all but one archimedean absolute places and add a non-archimedean absolute value $\nu_{\frak{p}}$ attached to a prime ideal $\frak{p}$ of $F$.
%over a rational prime $p$ which is totally split in $F$.
Thus, $\Sigma$ is of even cardinality and, by the well-known classification theorem of quaternion algebras, there exists a unique quaternion algebra $H$ up to isomorphism such that $H$ ramifies exactly for each $\nu\in \Sigma$. This is equivalent to say $H_{\nu}=H\otimes_{\mathbb{Q}}{F_{\nu}}=\mathrm{M}(2,F_{\nu})$ for all $\nu\not\in \Sigma$. Therefore $H$ splits only at one archimedean absolute value satisfying the condition \ref{cond-quat-alg}.

The corresponding result holds for $\Sigma'$ constructed from $\Sigma$ by adding an even number of non-archimedean absolute values as before. Therefore, there exists infinitely many quaternion algebras satisfying the desired condition.
\end{proof}

\subsection{Cyclotomic Fuchsian codes}

Let $p$ be an odd prime,  $\zeta_p\neq 1$ a primitive $p$-th root of unity, and consider the $p$-th cyclotomic field $\mathbb{Q}(\zeta_p)$, $[\mathbb{Q}(\zeta_p):\mathbb{Q}]=p-1$. Then, fix the number field $F=\mathbb{Q}(\zeta_p+\zeta_p^{-1})$,  which is the maximal totally real subfield of $\mathbb{Q}(\zeta_p)$. We have that $[F:\mathbb{Q}]=(p-1)/2$. This will allow us to construct infinitely many Fuchsian codes of rate at least $3(p-1)/2$ provided that we can find quaternion algebras $H$ over $F$ satisfying the required condition, namely $H\otimes_{\mathbb{Q}}\mathbb{R}=\mathrm{M}(2,\mathbb{R})\times\mathbb{H}^{(p-3)/2}$. Let us denote by $\Omega_F$ the set of all possible absolute values attached to $F$, archimedean or not.

The following technical lemma will be used in order to make easier the construction of quaternion algebras over the field $\mathbb{Q}(\zeta_p+\zeta_p^{-1})$ satisfying the condition \ref{cond-quat-alg} and to be more explicit in the description of such quaternion $F$-algebras over which we can construct Fuchsian codes.

\begin{lem}[\cite{milnecft}, 6.13]Let $F$ be a number field, $b\in F^*$ and $\Sigma\subseteq\Omega_F$ a finite subset  of even cardinality. Suppose that $b$ is a non-square element in $F_{\nu}$ for any $\nu\in \Sigma$. Then, there exists an element $a\in K^*$ such that the quaternion algebra $\left(\frac{a,b}{F}\right)$ splits precisely at the absolute values $\nu\not\in \Sigma$.
\label{wiese}
\end{lem}

\begin{prop}
For every prime $p\geq 5$, there exists $a\in F=\mathbb{Q}(\zeta_p+\zeta_p^{-1})$ such that
$$
\left(\frac{a,-1}{F}\right)\cong\mathrm{M}(2,\mathbb{R})\times\mathbb{H}^{\frac{p-3}{2}}.
$$
\label{infnite}
\end{prop}

\begin{proof}
Let us define a set $\Sigma$ of absolute values of places $F$ in the following way:
\begin{itemize}
\item If $p\equiv 3\pmod{4}$, take $\Sigma$ to be the set of all archimedean absolute values minus one.
\item If $p\equiv 1\pmod{4}$, take $\Sigma$ to be the set of all archimedean absolute values minus one, adjoining a finite place, $\nu_{\frak{q}}$, attached to a prime ideal $\frak{q}$ of $F$ over a rational prime $q$ such that
    \begin{itemize}
    \item[] $q\equiv 1\pmod{p}$,
    \item[] $q\equiv 3\pmod{4}$.
    \end{itemize}
    Notice that such a prime always exists, due to the Chinese remainder theorem. Indeed there are infinitely many, due to the theorem by Dirichlet on primes in arithmetic progressions.
\end{itemize}

In the first case, $\Sigma$ contains only archimedean absolute values. Since for any $\nu\in\Sigma$, $\nu(F)\subseteq\mathbb{R}$, it turns out that $-1$ is not a square in $F_{\nu}$, and by applying lemma \ref{wiese}, we have that there exists $a\in F$ such that $\left(\frac{a,-1}{F}\right)\cong\mathrm{M}(2,\mathbb{R})\times\mathbb{H}^{\frac{p-3}{2}}$.

In the second case, likewise, for all the archimedean places $\nu\in\Sigma$, $-1$ is not a square in $F_{\nu}$. The same hold for the the non-archimedean absolute value $\nu_{\frak{q}}$:

First, by a well known result on cyclotomic fields (cf. \cite{washington} 2.13),  $q$ factors in $p-1$ distinct prime ideals in  $\mathbb{Q}(\zeta_p)$. Hence, $q$ will factor in $(p-1)/2$ distinct prime ideals in $F$, i.e, is totally split. Hence, denoting by $\mathcal{O}_{F_{\nu_{\frak{q}}}}$ the ring of integers of $F_{\nu_{\frak{q}}}$, and (abusing notation), by $\frak{q}$ the unique maximal ideal of $\mathcal{O}_{F_{\nu_{\frak{q}}}}$, we have that $\mathcal{O}_{F_{\nu_{\frak{q}}}}/\frak{q}\cong\mathbb{F}_q$. Hence, if $-1$ were a square in $F_{\nu_{\frak{q}}}$ (hence in $\mathcal{O}_{F_{\nu_{\frak{q}}}}$), by reducing modulo $\frak{q}$, we would have that $\left(\frac{-1}{q}\right)=1$, a contradiction with the fact that $q\equiv 3\pmod{4}$.

This way, $-1$ is a non-square for each absolute value $\nu\in\Sigma$, and applying lemma \ref{wiese} again, the result holds.
\end{proof}

Applying Theorem \ref{thm-rate} and this proposition, we deduce the following corollary.

\begin{cor}
For any $p$ odd prime, such that $p\equiv 1\pmod{4}$, there exist infinitely many Fuchsian codes with rate at least $3(p-1)/2$, related to $F=\mathbb{Q}(\zeta_p+\zeta_p^{-1})$. The corresponding quaternion algebras can be taken of the form $\left(\frac{a,-1}{F}\right)$, for an infinite family of elements $a\in F^*$.
\end{cor}
\begin{proof}
In the proof of proposition \ref{infnite}, we have freedom to choose among infinitely many primes $q$ satisfying the two conditions
\begin{itemize}
\item[] $q\equiv 1\pmod{p}$,
\item[] $q\equiv 3\pmod{4}$.
\end{itemize}
For each of these primes $q$, we choose a prime ideal $\frak{q}$ above $q$, yielding a quaternion algebra $\left(\frac{a_q,-1}{F}\right)$, ramifying at $\frak{q}$ and at some archimedean primes. Since, due to the classification theorem, quaternion algebras ramifying at the same places are isomorphic, the elements $a_q$ have to be distinct all of them.

Now, since the ring of integers of $F$ is $\mathbb{Z}[\zeta_p+\zeta_p^{-1}]$, the Fuchsian code attached to each of the quaternion $F$-algebras $\left(\frac{a_q,-1}{F}\right)$ is defined to be image under the left regular representation of the natural order $\mathbb{Z}[\zeta_p+\zeta_p^{-1}][1,I,J,K]$, where $I^2=a_q$.
\end{proof}

\begin{rem}
For a fixed $F=\mathbb{Q}(\zeta_p+\zeta_p^{-1})$, the previous proposition provides a method for choosing quaternion algebras over $F$ in order to construct Fuchsian codes with rate $3(p-1)/2$. However, there are also other means to construct Fuchsian codes. For instance, for $p=7$, by using Magma, we see that the quaternion algebra over $F$ given by $a=b=\zeta_7+\zeta_7^{-1}$ also satisfies the conditions and leads to Fuchsian codes of rate at least $9$.
\end{rem}

\section{Arithmetic Fuchsian groups of signature $(1;e)$}

In this section, we give an explicit construction for a particular family of Fuchsian groups derived from the so called arithmetic Fuchsian groups of signature $(1;e)$.
Some of these Fuchsian groups were also  considered  in \cite{BHR}.
Here, we will deal with some Fuchsian groups derived from the arithmetic Fuchsian groups of signature $(1;e)$ defined over totally real fields. In the first subsection, we explicitly construct the Fuchsian codes attached to the arithmetic Fuchsian codes of signature $(1;e)$, while in the second subsection we provide numerical data to demonstrate that, at least in some example cases, higher code rate allows us to increase the codebook size (equivalently, data rate) for a fixed minimum distance and maximum transmission power (cf. Section 1). Future work consists of giving a rigorous proof for this fact.

\subsection{Explicit construction}

Arithmetic Fuchsian groups were characterized in \cite{tak75} by Takeuchi, who moreover classified and gave a complete list of them in the case of signature $(1;e)$ in \cite{tak}, determining the associated quaternion algebra up to isomorphism. We summarize below the main properties  useful for this paper. We refer to the same references for algebraic details and proofs.

\begin{prop} Let $T$ be an arithmetic Fuchsian group of signature $(1;e)$ associated to a division quaternion algebra $H$. We can assume $-\Id\in T$. Then,
\begin{enumerate}[i)]
\item The genus of the compact Riemmann surface $\mathcal{H}/T$ is 1.
\item There exist $\alpha, \beta, \gamma\in T$ satisfying $\Tr(\alpha), \Tr(\beta) >2$
      and
      $\Tr(\gamma)=2\cos \left(\frac{\pi}{e}\right)$,
      such that the group $T$  admits a presentation of the form
$$T=\left\langle \alpha,\beta, \gamma \mid  \alpha\beta\alpha^{-1}\beta^{-1} \gamma=-\Id, \, \gamma^e=-\Id \right\rangle .$$

\item A fundamental triple $(\alpha, \beta,\gamma)$ of generators of $T$ is uniquely determined by  $(x,y,z)=(\Tr(\alpha), \Tr(\beta),\Tr(\gamma))$, up to $\GL(2,\mathbb{R})$-conjugation.

\end{enumerate}
\end{prop}

\begin{prop}
\label{prop-T2} Consider a group $T$ as above, determined by the generators $\alpha, \beta,\gamma$ and  $(x,y,z)$ as above. Denote by  $T^{(2)}$ the subgroup of $T$ generated by  $\{\gamma^2 \mid \gamma \in T\}$. Then,
\begin{enumerate}[i)]
\item $T^{(2)}$ is a normal subgroup of $T$, and  $[T:\{\pm Id\}T^{(2)}]=4$.
\item $T^{(2)}=\langle \alpha^2,\beta^2,\gamma, \alpha\gamma\alpha^{-1}, \beta\gamma\beta^{-1}, \alpha\beta\gamma\beta^{-1}\alpha^{-1}\rangle$.
\item $T^{(2)}$ is a Fuchsian group derived from a quaternion algebra $\left(\frac{a,b}{F}\right)$, where $F=\mathbb{Q}(x^2,y^2,xyz)$, $a=x^2(x^2-4)$ and $b=-(2+2\cos(\pi/e)x^2y^2)$. In particular, $T^{(2)}$ is contained in the image by the regular representation of the group of units of reduced norm $1$ of a maximal order of $\left(\frac{a,b}{F}\right)$.
\end{enumerate}
\end{prop}

The generators for the groups $T^{(2)}$ will be made explicit by using the following result on the groups $T$, proved by Sijsling, cf. \cite{sij}.

\begin{prop}
\label{sijsling} \label{generatorsTak}
Let $T$ be an arithmetic Fuchsian group of signature $(1;e)$ generated by $\alpha$ and $\beta$. Then, after a change of variables, we can suppose that $\alpha=\left(\begin{array}{cc}\lambda\; & 0\\0\; & \lambda^{-1}\end{array}\right)$ with $\lambda$ an algebraic integer and $\beta=\left(\begin{array}{cc}a\; & b\\b\; & a\end{array}\right)$.
\end{prop}

We are interested in the explicit construction of codes attached to $\Gamma=T^{(2)}$ for several groups $T$ in the list given by Takeuchi. Table \ref{table1} gives the parameters of the sample groups we will consider, using  the same notation as in the  previous propositions.

\begin{table}[h]
\caption{Parameters of our sample groups}
\label{table1}
\resizebox{\columnwidth}{!}{%
\begin{tabular}{|c|c|l|l|c|c|}
\hline Group & $(x,y,z)$ & e & F & a & b\\
\hline $T_1$ & $\left(\sqrt{5},2\sqrt{3},\sqrt{15}\right)$ & $2$ & $\mathbb{Q}$ & 5 & -30\\
\hline $T_2$ & $\left(\sqrt{3+\sqrt{5}},\sqrt{9+3\sqrt{5}},\sqrt{6+\frac{9}{2}\sqrt{5}}\right)$ & $5$ & $\mathbb{Q}(\sqrt{5})$ & $2+2\sqrt{5}$ & $-6(25+11\sqrt{5})$\\
\hline $T_3$ & $\left(\sqrt{3+\sqrt{3}},\sqrt{8+4\sqrt{3}},\sqrt{9+5\sqrt{3}}\right)$ & $2$ & $\mathbb{Q}(\sqrt{3})$ & $2\sqrt{3}$ & $-3-2\sqrt{3}$\\
\hline $T_4$ & $\left(\sqrt{3+\sqrt{5}},\sqrt{6+2\sqrt{5}},\sqrt{7+3\sqrt{5}}\right)$ & $2$& $\mathbb{Q}(\sqrt{5})$ & $2+2\sqrt{5}$ & $-14-6\sqrt{5}$\\
\hline $T_5$ & $\left(\sqrt{2+(1+\sqrt{13})/2},\sqrt{16+4\sqrt{13}},\sqrt{12+ \frac{9}{2}(1+\sqrt{13})}\right)$ & $2$ & $\mathbb{Q}(\sqrt{13})$ & $(-1+\sqrt{13})/2$ & $-3(11+3\sqrt{13})$\\
\hline $T_6$ & $\left(\sqrt{3+\frac{1}{2}(1+\sqrt{5})},\sqrt{14+6\sqrt{5}},\sqrt{16+7\sqrt{5}}\right)$ & $5$ & $\mathbb{Q}(\sqrt{5})$ & $(-1+3\sqrt{5})$ & $-2(115+51\sqrt{5})$\\
\hline $T_7$ & $\left(\sqrt{3+\sqrt{3}},\sqrt{14+6\sqrt{3}},\sqrt{15+8\sqrt{3}}\right)$ & $6$ & $\mathbb{Q}(\sqrt{3})$ & $6(4+\sqrt{3})$ & $-3(31+18\sqrt{3})$\\
\hline
\end{tabular}%
}
\end{table}

For each group $\Gamma=T_i^{(2)}$, the code construction process can be made explicit.

First, we can easily find the generators of $T_i^{(2)}$, by applying Proposition \ref{sijsling} to compute the explicit matrices $\alpha$ and $\beta$. Namely, given the trace triple $(x,y,z)$, the equation $\mathrm{Tr}(\alpha)=\lambda+\lambda^{-1}=x$ will determine $\alpha$; then,
% writing $\beta=\left(\begin{array}{cc}a\; & b\\b\; & a\end{array}\right)$,
we obtain $a$ by solving $\mathrm{Tr}(\beta)=2a=y$; and finally, $\mathrm{Tr}(\alpha\beta)=z$ will give $b$, determining $\beta$. Thus, we have an explicit presentation of each group $\Gamma=T_i^{(2)}$.

The next step to construct our codes is to determine a fundamental domain for $\Gamma$ in each case. Fundamental domains for several arithmetic Fuchsian groups over $\mathbb{Q}$ can be found in \cite{alsinabayer}. For the general totally real case, it is more complicated. In the present work, we have effectively computed them with the aid of \emph{Mathematica}  by using the explicit generators of the groups computed above.

\begin{example} Let us consider $\Gamma=T_2^{(2)}$.
The generators of the group, given in Proposition \ref{prop-T2}, are obtained from:
$$
\begin{array}{l}
\alpha=\frac{1}{2}
\left(\begin{array}{cc}
\sqrt{3+\sqrt{5}}-\sqrt{-1+\sqrt{5}}\; & 0 \\
0\; & \sqrt{3+\sqrt{5}}+\sqrt{-1+\sqrt{5}}
\end{array}\right);\;
%\lambda=\frac{1}{2}\left(\sqrt{3+\sqrt{5}}-\sqrt{-1+\sqrt{5}}\right)\\
\\
\beta=\frac{1}{2}\left(\begin{array}{cc}\sqrt{3(3+\sqrt{5})}\;\; & -\sqrt{5+3\sqrt{5}}\\ &\\-\sqrt{5+3\sqrt{5}}\;\; & \sqrt{3(3+\sqrt{5})}\end{array}\right).
\end{array}
$$

A fundamental domain for $\Gamma=T_2^{(2)}$  is displayed in Figure \ref{Fdtal}. Its edges are given by the isometric circles of the following transformations:
$$
\begin{array}{l}
\alpha^2,\alpha^{-2},\beta^2,\beta^{-2},\gamma,\gamma^{-1},
\alpha^{-1}\beta\alpha\beta^{-1},\left(\alpha^{-1}\beta\alpha\beta^{-1}\right)^{-1},
\\
\alpha\beta^{-1}\alpha^{-1}\beta,\left(\alpha\beta^{-1}\alpha^{-1}\beta\right)^{-1},
\alpha^{-1}\beta^{-1}\alpha\beta,\left(\alpha^{-1}\beta^{-1}\alpha\beta\right)^{-1}.
\end{array}
$$
\end{example}

    \begin{figure}
            \begin{center}
                \scalebox{1}{
                        \includegraphics[width=0.45\textwidth]{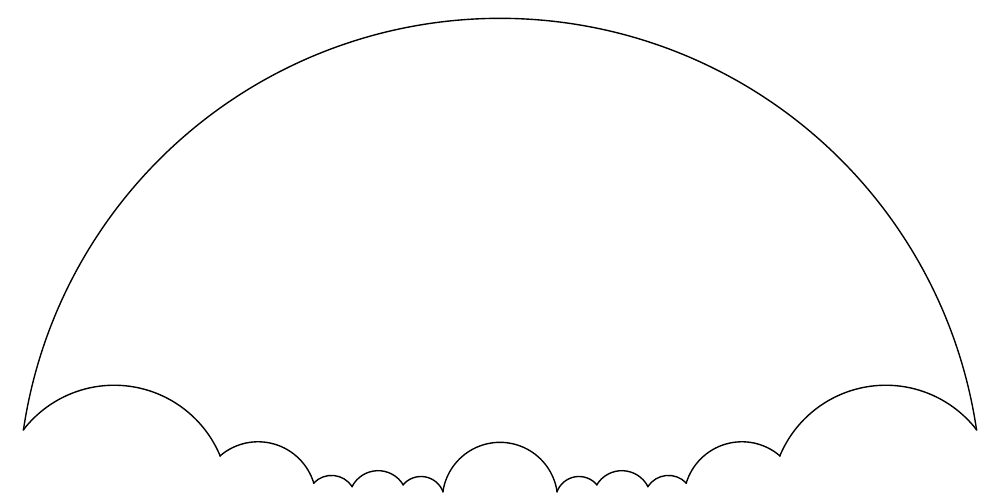}
                }
                        \caption[]{Fundamental domain for group $T_2^{(2)}$ }
                        \label{Fdtal}
        \end{center}
    \end{figure}

Now we turn the attention to the design of the codes. We fix $\tau=i$ as the center of the code, and for each code size $q=2N$, we have chosen a set of matrices $S_{\Gamma}=\{\gamma_1,...,\gamma_N\}$ with the aim of minimizing the average energy $\frac{1}{N}\sum_{i=1}^N|\gamma_i(\tau)|^2$.
We display in tables \ref{table4cons} and \ref{table16cons} the choices for the set of elements $S_{\Gamma}$ for the sample groups $\Gamma=T_i^{(2)}$, giving 4-NUF and 16-NUF constellations.

\begin{table}[!h]
\caption{Choices for the $4$-NUF codes}
\label{table4cons}
\begin{tabular}{|c|l|}
\hline Group $T_i$ & $S_{\Gamma}$, for $\Gamma=T_i^{(2)}$ \\
\hline $T_1$ & $\pm\alpha^2$, $\pm\alpha\gamma\alpha^{-1}$\\
\hline $T_2$ & $\pm\alpha^2$, $\pm\alpha\gamma\alpha^{-1}$\\
\hline $T_3$ & $\pm\alpha^2$, $\pm\alpha\gamma\alpha^{-1}$\\
\hline $T_4$ & $\pm\alpha^2$, $\pm\alpha\gamma\alpha^{-1}$\\
\hline $T_5$ & $\pm Id$, $\pm\alpha^2$\\
\hline $T_6$ & $\pm\alpha^2$, $\pm\alpha\gamma\alpha^{-1}$\\
\hline $T_7$ & $\pm\alpha^2$, $\pm\alpha\gamma\alpha^{-1}$\\
\hline
\end{tabular}
\end{table}

\begin{table}[!h]
\resizebox{0.8\textwidth}{!}{\begin{minipage}{\textwidth}
\caption{Choices for the $16$-NUF codes}
\label{table16cons}
\begin{tabular}{|c|l|}
\hline Group $T_i$ & $S_{\Gamma}$, for $\Gamma=T_i^{(2)}$ \\
\hline $T_1$ & $\pm\alpha^4$, $\pm\alpha^2\gamma$, $\pm\alpha^2\beta^2$, $\pm\alpha\gamma\alpha^{-1}\gamma$, $\alpha^3\gamma\alpha^{-1}$, $\pm\alpha^2\beta\gamma\beta^{-1}$, $\pm\alpha\gamma\beta\gamma\beta^{-1}\alpha^{-1}$, $\pm\alpha^3\beta\gamma\beta^{-1}\alpha^{-1}$\\
\hline $T_2$ & $\pm\alpha^2$, $\pm\alpha^4$, $\pm\alpha^2\gamma$, $\pm\alpha^2\beta^2$, $\pm\alpha\gamma\alpha^{-1}$, $\pm\alpha\gamma\alpha$, $\pm\alpha\gamma\alpha^{-1}\beta^2$, $\pm\alpha^3\beta\gamma\beta^{-1}\alpha^{-1}$ \\
\hline $T_3$ & $\pm\alpha^2$, $\pm\alpha^4$, $\pm\alpha^2\gamma$, $\pm\alpha^3\gamma\alpha^{-1}$, $\pm\alpha^2\beta\gamma\beta^{-1}$, $\pm\alpha\gamma\alpha^{-1}\beta^2$, $\pm\alpha\gamma\beta\alpha^{-1}\beta^{-1}$, $\pm\alpha^3\beta\gamma\beta^{-1}\alpha^{-1}$\\
\hline $T_4$ & $\pm\alpha^2$, $\pm\alpha^4$, $\pm\alpha^2\gamma$, $\pm\alpha\gamma\alpha$, $\pm\alpha^3\gamma\alpha^{-1}$, $\pm\gamma\beta\gamma\beta^{-1}$,  $\pm\alpha\gamma\alpha^{-1}\beta^2$, $\pm\alpha^3\beta\gamma\beta^{-1}\alpha^{-1}$\\
\hline $T_5$ & $\pm\alpha^2$, $\pm\alpha^4$, $\pm\alpha^2\gamma$, $\pm\alpha\gamma\alpha^{-1}$, $\pm\alpha^3\gamma\alpha^{-1}$, $\pm\alpha\gamma\alpha^{-1}\beta^2$, $\pm\alpha\gamma\alpha^{-1}\gamma$, $\pm\beta\gamma\beta^{-1}\alpha^2$\\
\hline $T_6$ & $\pm\alpha^2$, $\pm\alpha^4$, $\pm\alpha^2\gamma$, $\pm\alpha\gamma\alpha^{-1}$, $\pm\alpha^3\gamma\alpha^{-1}$, $\pm\alpha\gamma\alpha^{-1}\gamma$, $\pm\alpha^2\beta\gamma\beta^{-1}$, $\pm\alpha^3\beta\gamma\beta^{-1}\alpha^{-1}$\\
\hline $T_7$ & $\pm\alpha^2$, $\pm\alpha^4$, $\pm\alpha\gamma\alpha^{-1}$, $\pm\alpha^2\gamma$, $\pm\alpha\gamma\alpha^{-1}\gamma$, $\pm\alpha^3\gamma\alpha^{-1}$, $\pm\alpha^2\beta\gamma\beta^{-1}$, $\pm\alpha^3\beta\gamma\beta^{-1}\alpha^{-1}$\\
\hline
\end{tabular}
\end{minipage}}
\end{table}

Given the choice of the center $\tau$ and $S_{\Gamma}$, the associated Fuchsian code is $\mathcal{C} = \left\{\pm\gamma(\tau) \, \mid \,  \gamma \in S_{\Gamma}\right\}\subseteq \mathbb{C}$, taking in account the duplication to the lower half-plane. For instance,
Figure \ref{4cons} depicts a 4-NUF constellation for the group $T_2^{(2)}$.

       \begin{figure}
        \begin{center}
                \scalebox{1.35}{
                        \includegraphics[width=0.45\textwidth]{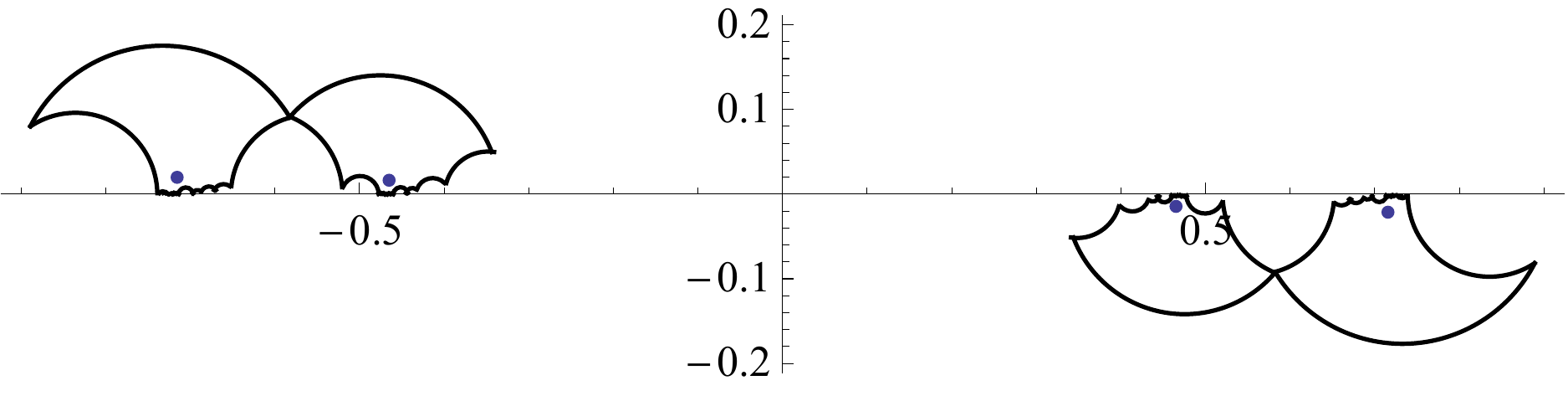}
                }
                        \caption[]{$4$-NUF constellation and tiles for group $T_2^{(2)}$}\label{4cons}
        \end{center}
    \end{figure}

\begin{rem} In Takeuchi's list, the arithmetic Fuchsian groups of signature $(1;e)$ are defined over number fields of degree $1,2,3,4,5$ and $6$. Hence, in an analogous way Fuchsian codes with data rates $3,6,9,12,15$ and $18$ can be explicitly constructed.
\end{rem}

\subsection{Numerical data}
Next, we show numerically in the case of Fuchsian codes attached to arithmetic Fuchsian groups of signature $(1;e)$, how a bigger code rate implies a bigger number of codewords with prescribed minimum distance fitting in a Euclidean ball. Hence, in this case, a higher code rate does indeed give us higher data rate without any penalty in the minimum distance or transmission power, along the same lines as for lattice codes, cf. Section 1.

Our approach is as follows: recall that the groups $\Gamma=T_2^{(2)}$ are generated by the elements $\alpha^2$, $\beta^2$, $\gamma$, $\alpha\gamma\alpha^{-1}$, $\beta\gamma\beta^{-1}$ and $\alpha\beta\gamma\beta^{-1}\alpha^{-1}$. Using \emph{Mathematica}, we have implemented an algorithm which generates a set of matrices consisting in: first, all the generators of $\Gamma$, second, all the possible matrices of the form $\gamma_1\gamma_2$, where $\gamma_1,\gamma_2$ runs over the set of generators, and then, we iterate this procedure until we have a desired number of matrices fixed beforehand. After that, our algorithm deletes the possible repeated matrices of the former stage. This way, we can construct a codebook of each desired size.

 Once our codebook is constructed, we count how many codewords are there in the unit ball and in the ball of radius $0.5$ for each of the groups $T_i$. The center of the code has been taken to be $i$. The minimum distance $d$ (computed up to an accuracy of four decimal digits), in each case, is that of the bigger codebook, i.e., the one fitting inside the unit ball. Notice that, among our choice of codes, the minimum distance corresponds to a code of rate $6$.

\begin{table}[!h]
\caption{Number of codewords inside Euclidean balls of radius $r=1,0.5$.}
\label{tablerates1}
%\resizebox{\columnwidth}{!}{%
\begin{tabular}{|c|c|c|c|c|}
\hline Group & Rate &  $r=1$    &   $r=0.5$ & $d$\\
\hline $T_1$ &  $3$   & $72$ & $36$ & $0.0001$\\
\hline $T_2$ &  $6$ & $84$   & $54$ & $0.0005$\\
\hline $T3$ & $6$ & $80$ & $42$ & $0.0002$\\
\hline $T4$ & $6$ & $80$ & $40$ & $0.0001$\\
\hline $T5$ & $6$ & $80$ & $40$ & $0.0003$\\
\hline $T6$ & $6$ & $84$ & $42$ & $0.0002$\\
\hline $T7$ & $6$ & $84$ & $42$ & $0.0002$\\
\hline
\end{tabular}%
\end{table}

\section{Conclusions}

In this paper,  we have generalized the construction of the Fuchsian codes presented in \cite{BHR} and \cite{BHAR} to the general case of Fuchsian groups over totally real fields. One of the main features of the codes in \cite{BHR,BHAR} is that they have logarithmic decoding complexity. We have shown that by adapting the point reduction algorithm introduced in \cite{bayerremon} to the present, more general case, the decoding complexity of the corresponding Fuchsian codes remains logarithmic in the codebook size, provided that we have a fundamental domain and a representation of the Fuchsian group.

In the case of Fuchsian groups $\Gamma$ associated  to quaternion algebras defined over a totally real number field, with an additional hypothesis about their ramification, we have proved that the Fuchsian code attached to $\Gamma$   has  rate at least $3n$, where  $n$ is the degree of the base field. This corresponds to at least $3n$-fold (resp. $3n/2$-fold) information compression in terms of the number of independent integers transmitted per codeword compared to the commonly used PAM (resp. QAM) alphabet.  Moreover,  we have deduced that there exist infinitely many Fuchsian codes of  rate $3n$. In particular, by considering  subfields of cyclotomic fields, we have made explicit the existence of infinitely many Fuchsian codes with  rate $3\frac{p-1}{2}$. We have explicitly constructed Fuchsian codes attached to the groups $T^{(2)}$ for arithmetic Fuchsian groups $T$ of signature $(1;e)$ classified by Takeuchi. Finally, the relevance of the code rate in terms of information compression and data rate has been numerically demonstrated.

Further research will consist of rigorously finding and proving a relation between the code rate and the data rate, and of  the construction of an error-correcting outer system for our codes. Possible enabler of this is the excess of code rate that could be alternatively utilized for error correction, e.g., by using some kind of an analogy of a parity-check method. This would make our codes more suitable  for the low-moderate SNR regime (SNR stands for the signal-to-noise ratio describing the channel quality), while at the moment the relevant application is the high-SNR regime.  One instance of this is  an optic-fiber channel.


\begin{thebibliography}{99}

\bibitem{alsinabayer} Alsina, M.; Bayer, P.: \emph{Quaternion orders, quadratic forms and Shimura curves.}
CRM Monograph Series, 22. American Mathematical Society, Providence, RI, 2004. xvi+196 pp. ISBN: 0-8218-3359-6.

\bibitem{BHAR} Blanco-Chac\'{o}n, I.; Rem\'{o}n, D.; Hollanti, C.; Alsina, M.: Nonuniform Fuchsian codes for noisy channels. \emph{Journal of the Franklin Institute} 351 (2014) 5076--5098.

\bibitem{bayerremon} Bayer, P.; Rem\'on, D.: A reduction point algorithm for cocompact Fuchsian groups and applications. Adv. Math. Commun. 8 (2014) 223--239.

\bibitem{BHR}Blanco-Chac\'{o}n, I.; Hollanti, C.; Rem\'{o}n, D.: Fuchsian codes for AWGN channels. PREPROCEEDINGS. The International Workshop on Coding and Cryptography, WCC 2013. p. 496--507. Bergen (2013). ISBN: 978-82-308-2269-2.
\bibitem{dvb} Digital Video Broadcasting Consortium, \emph{dvb.org}.

\bibitem{brazilian_COAM}
Carvalho, E., Andrade, A., Palazzo, R., Filho, J.V.: Arithmetic {Fuchsian}
  groups and space--time block codes.
Comput. Appl. Math. \textbf{30}, 485--498 (2011)



\bibitem{59paper}
Gertsenshtein, M., Vasilev, V.: Waveguides with random inhomogeneties and
  {Brownian} motion in the {Lobachevsky} plane.
Theory Probab. Appl. \textbf{4}, 391--398 (1959)


\bibitem{maxorder}
Hollanti, C., Lahtonen, J.: A new tool: Constructing {STBC}s from maximal
  orders in central simple algebras.
 In: IEEE Information Theory Workshop (ITW '06), Punta del Este,
  Uruguay, pp. 322--326 (2006)

\bibitem{johansson} S. Johansson: On Fundamental Domains of Arithmetic Fuchsian Groups. \emph{Math. Comp.} 69 (2000), no. 229, 339--349.

\bibitem{katok} Katok, S.: \emph{Fuchsian Groups.} Chicago Lectures in Mathematics Series. The University of Chicago  Press (1992).

\bibitem{milnecft} J. S. Milne: \emph{Class field theory (4.02)} \hyperref[http://www.jmilne.org/math/CourseNotes/cft.html]{http://www.jmilne.org/math/CourseNotes/cft.html} (2013)

\bibitem{OV} Oggier, F.; Viterbo, E.: Algebraic number theory and code design for Rayleigh fading channels. \emph{Commun. Inf. Theory} 1(3) (2004), 333--416.

\bibitem{SRS}
Sethuraman, B.A., Rajan, B., Shashidhar, V.: Full-diversity, high-rate
  space--time block codes from division algebras.
IEEE Transactions on Information Theory \textbf{49}(10), 2596--2616
  (2003)

\bibitem{shimura1967} G. Shimura: \emph{Construction of class fields and zeta functions of algebraic curves},
 Annals of Math., 85, 1967, 58-159.


\bibitem{sij} J. Sijsling: \emph{Equations for arithmetic pointed tori}. Ph.D. Thesis, Universiteit Utrecht, 2010.

\bibitem{brazilian_franklin2}
da~Silva, E.B., Firer, M., Costa, S.R., Palazzo, R.: Signal constellations in
  the hyperbolic plane: A proposal for new communication systems.
 Journal of the Franklin Institute \textbf{343}, 69--82 (2006)


\bibitem{brazilian_franklin}
de~Souza, M., Faria, M.B., Palazzo, R., Firer, M.: Edge-pairing isometries and
  counting dirichlet domains on the densest tessellation (12g-6,3) for signal
  set design.
Journal of the Franklin Institute \textbf{349}, 1139--1152 (2012)


\bibitem{tak75} Takeuchi, K.:  A characterization of arithmetic Fuchsian groups. J. Math. Soc. Japan, 27, Number 4, 1975, 600-612.

\bibitem{tak} Takeuchi, K.: Arithmetic Fuchsian groups with signature $(1,e)$. J. Math. Soc. Japan, 35, Number 3, 1983, 381-407.

\bibitem{brazilian_IEEE}
Vieira, V.L., Palazzo, R., Faria, M.B.: On the arithmetic {Fuchsian} groups
  derived from quaternion orders.
\newblock Proceedings of the International Telecommunications Symposium (ITS
  2006), Fortaleza-Ce (Brazil)  (2006)

\bibitem{vigneras} Vign\'{e}ras, M.\,F.: \emph{Arithm\'{e}tique des alg\`{e}bres de quaternions}.
Lecture Notes in Mathematics 800. Springer, 1980. vii+169 pp. ISBN: 3-540-09983-2.

\bibitem{voight} J. Voight: Computing fundamental domains for Fuchsian groups. \emph{J. Th\'{e}orie des Nombres de Bordeaux} 21 (2009), 467--489.

\bibitem{washington} L.C. Washington: \emph{Introduction to cyclotomic fields (second edition)}. Springer GTM, Number 83 (1982).
\end{thebibliography}
\end{document}